\documentclass[conference,10pt]{IEEEtran}
\usepackage{cite}

\ifCLASSINFOpdf
   \usepackage[pdftex]{graphicx}

\else

\fi

\bstctlcite{IEEEexample:BSTcontrol}

\usepackage{cite}

\usepackage{amsmath,amssymb,amsthm}

\usepackage{mathtools}

\usepackage{dblfloatfix}
\usepackage[caption=false,font=footnotesize]{subfig}
\usepackage{mathrsfs}

\newtheorem{theorem}{Theorem}

\newtheorem{lemma}[theorem]{Lemma}
\newtheorem{claim}{Claim}

\begin{document}

\title{Sum-networks: Dependency on Characteristic of the Finite Field under Linear Network Coding}
\author{\IEEEauthorblockN{Niladri Das and Brijesh Kumar Rai}
\IEEEauthorblockA{Department of Electronics and Electrical Engineering\\ Indian Institute of Technology Guwahati, Guwahati, Assam, India\\
Email: \{d.niladri, bkrai\}@iitg.ernet.in}}
\maketitle

\begin{abstract}
Sum-networks are networks where all the terminals demand the sum of the
symbols generated at the sources. It has been shown that for any finite set/co-finite set of
prime numbers, there exists a sum-network which has a vector linear solution if and only if the characteristic of the finite field belongs to the given set. It has also been shown that for any
positive rational number $k/n$, there exists a sum-network which has
capacity equal to $k/n$. It is a natural question whether, for any
positive rational number $k/n$, and for any finite set/co-finite set of primes
$\{p_1,p_2,\ldots,p_l\}$, there exists a sum-network which has a capacity
achieving rate $k/n$ fractional linear network coding solution if and only
if the characteristic of the finite field belongs to the given set. We show that indeed there exists such a sum-network by constructing such a sum-network. 
\end{abstract}

\section{Introduction}

The concept of network coding originated in the year 2000 \cite{alswede}. The area of network coding has been an area of active research since then. It was shown that network coding is necessary to achieve the min-cut bound in a multicast network. It was further shown, in subsequent works \cite{li, jaggi, medrad}, that linear network coding is sufficient to achieve the capacity of a multicast network. However, in a more general scenario where there are multiple sources and multiple terminals in a network, and each terminal requires information generated at a subset of sources, capacity cannot always be achieved through linear network coding \cite{doug}. In \cite{doug}, the authors constructed a network in which the capacity is achievable through non-linear network coding, but not through linear network coding. However, since linear network coding has an inherent structure to it, and as it is easier to implement practically as opposed to non-linear network coding, the possibility of achieving a greater throughput by using non-linear network codes does not undermines the necessity to study linear network coding. 

Linear coding capacity is defined as the supremum of all achievable rates that can be achieved by using linear network coding. It has been shown that linear coding capacity may depend on the characteristic of the finite field. The Fano network presented in \cite{doug} has linear coding capacity equal to one when the characteristic of the finite field is two, but when the characteristic is not two, the linear coding capacity is $4/5$. Another network, presented in the same reference \cite{doug}, has linear coding capacity equal to one if the characteristic of the finite field is not two, but has linear coding capacity equal to $10/11$ if the characteristic of the finite field is two. Later, in \cite{rai2}, it was shown that for any given finite or co-finite set of prime numbers, there exists a network which has linear coding capacity equal to one if and only if the characteristic of the finite field belong to the given set.

In this paper, we consider linear network coding in the context of a function computation problem over a network. We consider a special case of general function computation problem over a network where each terminal demand the sum of symbols generated at all the sources. Such a network has been termed as a sum-network in the literature. The study of sum-networks include the function computation problem where all the terminals demand a linear combination of source symbols when the source symbols are from a finite field \cite{rai2}. 

There have been numerous studies on sum-networks. In \cite{rama}, it has been shown that if the number of sources or the number of terminals in a sum-network is two and there is at least one path from each source to each terminal, then the linear coding capacity of such a sum-network is at least $1$. However, when the number of sources or terminals are greater than two, at least one path from each source to each terminal does not suffice achievability of rate $1$. References \cite{rai1} and \cite{lang} independently presented a sum-network having three sources and three terminals which has capacity equal to $\frac{2}{3}$. In reference \cite{lang}, the authors also showed that if the sum-network having three source and three terminal has min-cut equal to two between any source and any terminal, then the linear coding capacity of the network is at least $1$. However, this property also fails to hold for larger sum-networks \cite{rai5}. For a class of sum-networks having $3$ sources and $n$ ($\geq 3$) terminals or $m$ ($\geq 3$) sources and $3$ terminals, it was shown that there exist networks which have linear coding capacity of the form $\frac{k}{k+1}$, where $k$ is a positive integer greater than or equal to $2$ \cite{rai4}. It was also conjectured that the capacity of sum-network in this class is always of the form $\frac{k}{k+1}$ \cite{rai4}. Till date this conjecture has neither been proven nor anyone has given a counter example. 

The work in \cite{rai2} relates the study of sum-networks with the networks  where every terminal demands a subset of information generated at the sources, and thereby proved various results for the sum-networks. It was shown that for any multiple-unicast network there exists a sum-network which is solvable if and only if the the multiple unicast network is solvable. This showed that many results like insufficiency of liner network coding, unachievability of network coding capacity, nonreversibility, which are true from multiple unicast networks, are also true from sum-networks.  Moreover, given a sum-network, it has also been shown that there exists a linear solvably equivalent multiple-unicast network. 

For sum-networks too, the linear coding capacity is dependent on the characteristic of the finite field. Reference \cite{rai1} and \cite{rai2} show that for any given finite or co-finite set of primes, there exists a sum-network which has linear coding capacity equal to one if and only if the characteristic of the finite field belong to the given set. It follows from this results that, to achieve the capacity by using linear network coding, it is necessary to know over which finite field the network code must be designed. However, to the best of our knowledge, for sum-networks, the dependency of linear coding capacity on the characteristic of the finite field, has been shown only for sum-networks which has capacity equal to one. We generalize this result for sum-networks with capacity equal to any positive rational number. Specifically we show that, for any given positive rational number $k/n$, and for any given set of finite or co-finite set of prime numbers, there exists a sum-network which has capacity equal to $k/n$, where the rate $k/n$ can be achieved using fractional linear network coding if and only if the characteristic of the finite field belongs to the given set. We note that, in a recent work, similar results have been shown for the class of multiple-unicast networks \cite{das3}.

It is known that the capacity of sum-networks can be any rational number \cite{rai3}. For any given positive rational number $k/n$, Rai \textit{et al.} \cite{rai3} constructed a sum-network which has capacity equal to $k/n$. Reference \cite{tripathy} and \cite{das} improved this result by showing that for some rational numbers, the same result can be obtained in a sum-network that uses significantly less number of sources and terminals. Given the fact that one construction of sum-networks, which are linearly solvable if and only if the characteristic of the finite field belong to a finite/co-finite set of primes, has already been given in \cite{rai1}, it is intuitive to think that a combination of the sum-networks presented in \cite{rai1} and \cite{rai3} would lead to the result we desire in this paper. However, we have found that this is not so imminent. For example, in Fig.~\ref{sum1} we reproduce a sum-network from \cite{rai1}. This network, denoted here by $\mathcal{N}_a$, has a vector linear solution for any vector dimension if and only if the characteristic of the finite field is $2$. In Fig.~\ref{sum2} we show a sum-network $\mathcal{N}_b$ which has network $\mathcal{N}_a$ as a sub-network. It can be shown that the capacity of the network $\mathcal{N}_b$ is $\frac{3}{5}$. We find that the sub-network $\mathcal{N}_a$ fails to enforce the restriction on the characteristic of the finite field in sum-network $\mathcal{N}_b$. As shown in Fig.\ref{sum3}, the rate $\frac{3}{5}$ is achievable in $\mathcal{N}_b$ using fractional linear network codes over any finite field.

\begin{figure*}
\centerline{\subfloat[Sum-network $\mathcal{N}_a$: vector linearly solvable if and only if characteristic of the finite field is $2$.]{\includegraphics[width
=0.3\textwidth]{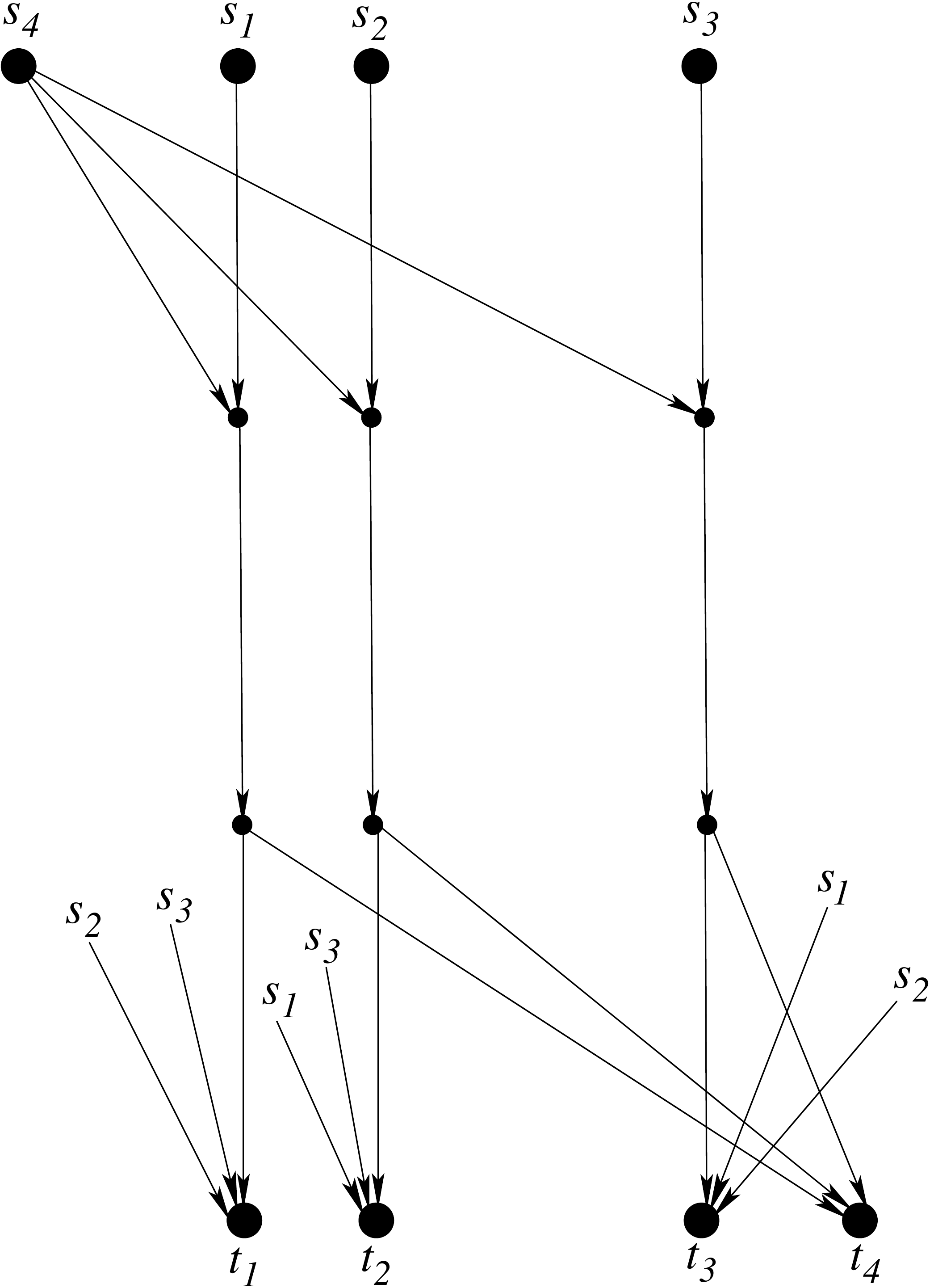}
\label{sum1}}
\hfil
\subfloat[Sum-network $\mathcal{N}_b$: The capacity reduces to $\frac{3}{5}$ by adding source $s_5$ and terminal $t_5$ to $\mathcal{N}_1$.]{\includegraphics[width=0.3\textwidth]{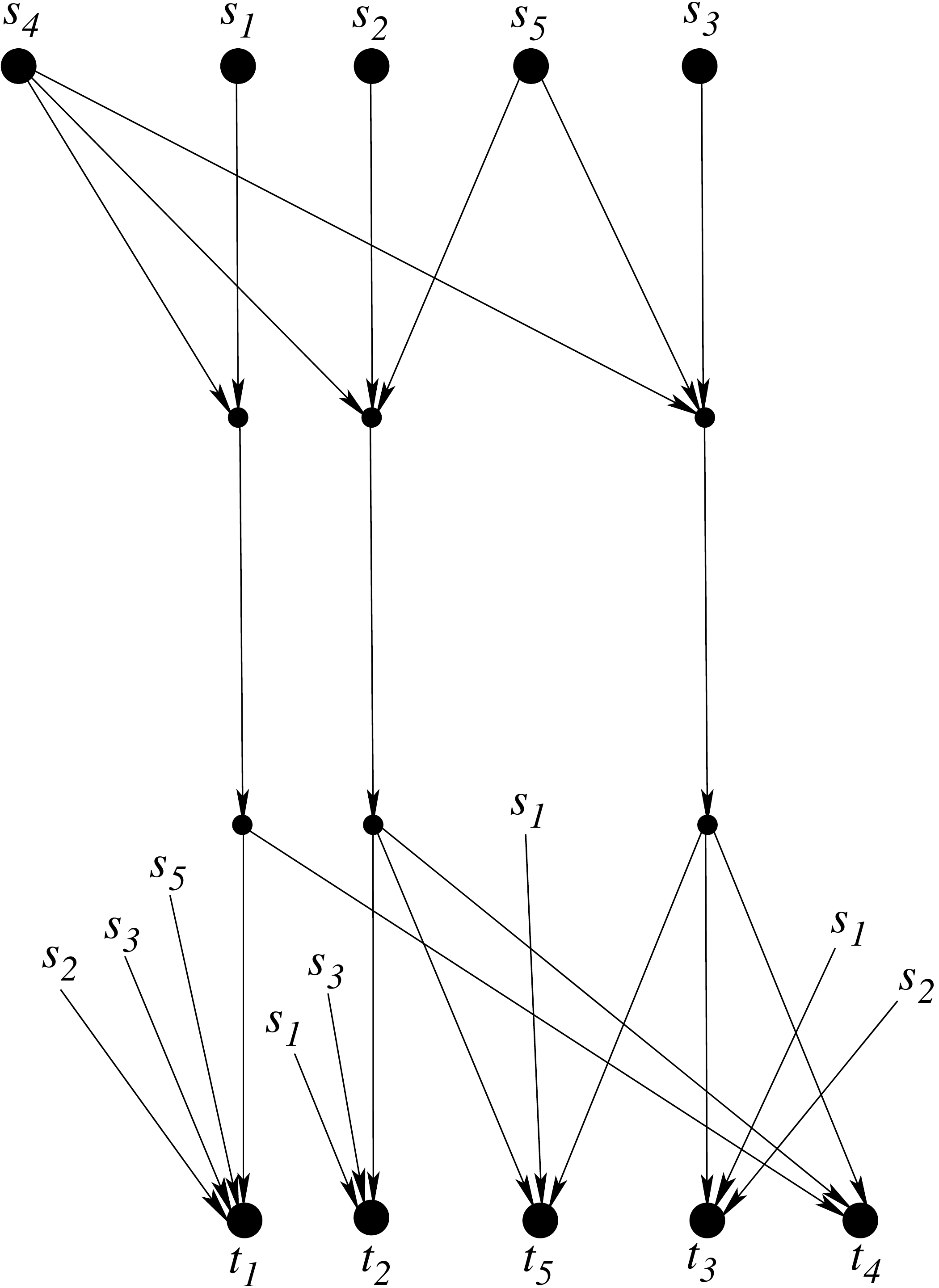}
\label{sum2}}
\hfil
\subfloat[The capacity $\frac{3}{5}$ of $\mathcal{N}_b$ is achievable by rate $\frac{3}{5}$ linear network code over all finite fields.]{\includegraphics[width=0.3\textwidth]{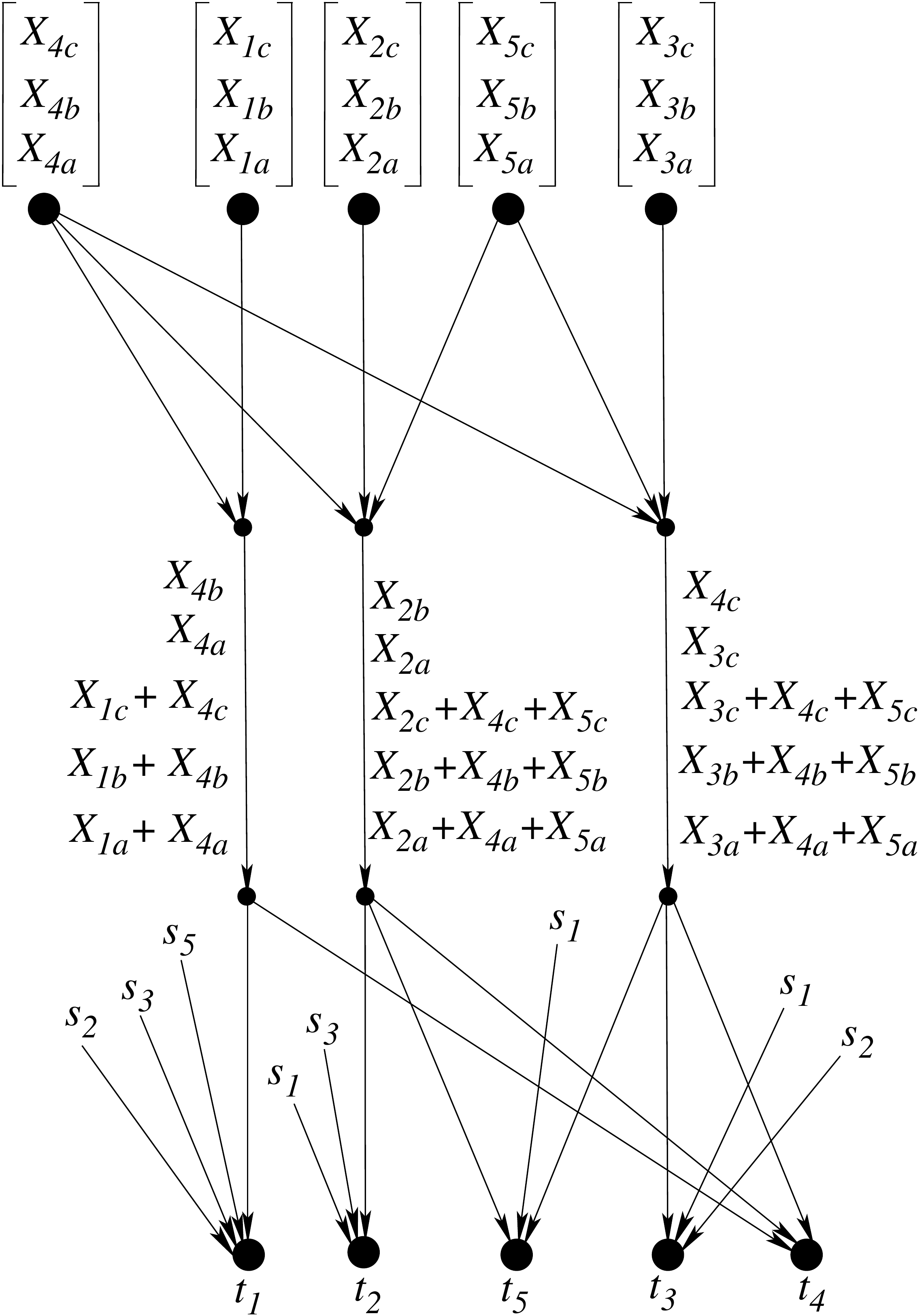}
\label{sum3}}}
\caption{The above sum-networks demonstrate that adding more sources and terminals to the sum-networks, which are linearly solvable if and only if the characteristic of the finite field is a certain number as shown in \cite{rai1} and \cite{rai2}, might remove the restriction on the characteristic of the field when the capacity gets reduced. An edge which is without a head represents an edge which emanates from the source as indicated on top of it. Sources are represented as a 3-length vectors, and the first, second and third component of a source message $X_i$ are denoted by $X_{ia},X_{ib}$ and $X_{ic}$ respectively. In Fig.~\ref{sum3} we have shown only the non-trivial part of the linear network code. The massage carried over three middle edges are shown alongside the edges.}
\label{sum}
\end{figure*}

In \cite{tripathy}, the authors presented two methods to construct sum-networks starting from graphs. In this paper, they made a remark (Remark 2), ``In Construction 2, we chose to connect the starred source only a carefully chosen subset of the bottleneck edges. If instead we had connected it to all the bottleneck edges (for instance), the starred terminal could only recover the source, under conditions on the characteristic of the field $\mathcal{A}$. Our choice of the subset of bottleneck edges avoids this dependence on the field characteristic.'' However, the authors have neither proven that such is the case when the starred source is connected to all of the bottleneck edges, nor they have shown any such example. We too have failed to construct any  such sum-networks using the construction methods given in \cite{tripathy}. For example, we consider the sum-network which was originally used as an example in \cite{tripathy} to demonstrate workings of the construction methods. We reproduce the sum-network here in Fig.~\ref{exp} but with the starred source connected to all of the bottleneck edges. This sum-network, as shown in \cite{tripathy} has capacity $2/5$. The Fig.~\ref{exp} shows that the sum-network has a rate $2/5$ achieving fractional linear network coding solution over all finite fields.
\begin{figure*}
\centering
\includegraphics[width=0.6\textwidth]{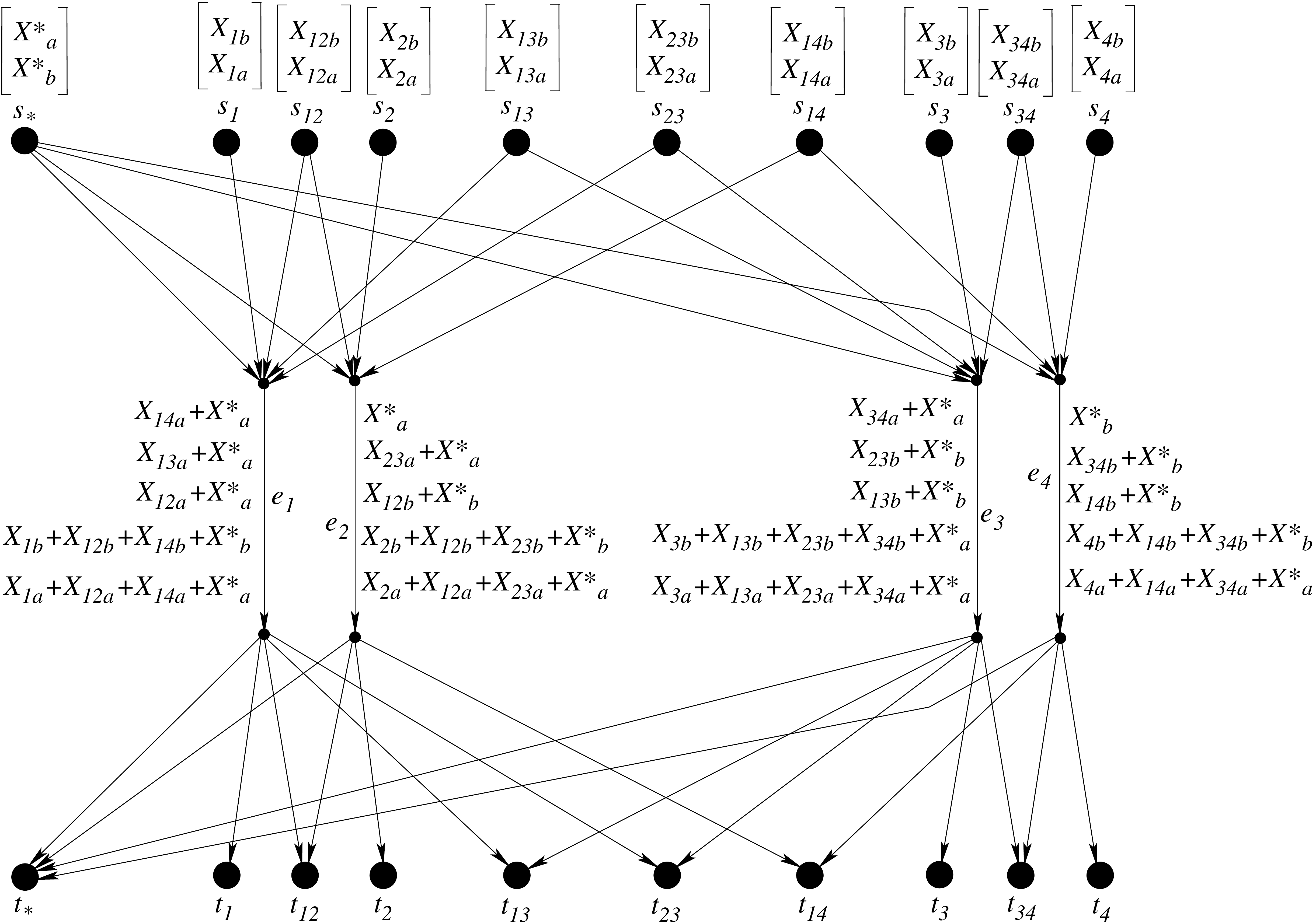}
\label{exp}
\caption{This sum-network is a presented with a modification of a sum-network presented in \cite{tripathy}. The modification is that in the original presentation the source $s_*$ has a path to the tail nodes of the edges $e_1,e_2$ and $e_3$; here $s_*$ additionally has a path to the tail of edge $e_4$. The capacity of this sum-network can be derived from \cite{tripathy}, and it is equal to $\frac{2}{5}$. The coding scheme shown here achieves the $2/5$ rate over all finite fields. Each source symbol is represented by a two length vector where the first and the second symbol is distinguished by the subscript $a$ and $b$ respectively. The massages sent over the edges $e_1,e_2,e_3$ and $e_4$ are shown beside the edges. The direct edges are not shown in the figure for the sake of clarity. Between any source and any terminal if there is no path in the figure, then it is assumed that there is a direct edge directed from that source to that terminal. It can be shown that by the help of these transmissions, along with the information coming through the direct edges, all the terminals compute the required sum.}
\end{figure*}

The remaining paper is divided into four sections. Section~\ref{sec1} contains the standard definitions of various terms used in this paper. In Section~\ref{sec2}, we show that for any set of primes $\{p_1,p_2,\ldots,p_w\}$, and any positive rational number $k/n$ there exists a sum-network which has a rate $k/n$ achieving fractional linear network coding solution if and only if the characteristic of the finite field belongs to the given set. In Section~\ref{sec3}, we show that for any set of primes $\{p_1,p_2,\ldots,p_w\}$, and any positive rational number $k/n$ there exists a sum-network which has a rate $k/n$ achieving fractional linear network coding solution if and only if the characteristic of the finite field does not belong to the given set. The paper is concluded in Section~\ref{sec4}.

\section{Preliminaries}\label{sec1}
A network is represented by a graph $G(V,E)$. The set of nodes $V$ is partitioned into three disjoint sets: $S,T$ and $V^\prime$. $S$ is the set of sources. The sources are assumed not to have any incoming edge. $T$ is the set of terminals which are assumed not to have any outgoing edge. All other nodes in $V$ which are neither a source node nor a terminal node are in the set $V^\prime$. The elements of $V^\prime$ are called the intermediate nodes. $(u,v,i)$ denotes the $i^\text{th}$ edge between the nodes $u$ and $v$. In case of a single edge between two nodes $u$ and $v$, the edge is represented by $(u,v)$. For an edge $e=(u,v,i)$, $head(e)$ represents the node $u$ and $tail(e)$ represents the node $v$. For any node $v\in V$, $In(v)$ denotes the set $\{e\in E | head(e) = v \}$. For any edge $e\in E$, the information carried by $e$ is denoted by $Y_{e}$. Each source $s_i\in S$ generates an i.i.d random process $X_i$ uniformly distributed over a finite alphabet. Any source process is independent of all other source processes. In a sum-network all terminals demand the same symbol, which is the sum of the symbols generated at the sources, i.e.\ $\sum_{i=1}^{|S|} X_i$.

A network code is assignment of edge functions, one for every edge; and decoding functions for the terminals, one for every terminal. An $(r,l)$ fractional network code over an alphabet $\mathcal{A}$ is described as follows. For any edge $e$ emanating from a source, there is a function $f_e: \mathcal{A}^r \rightarrow \mathcal{A}^l$. For any edge $e$ emanating from an intermediate node, there is a function $f_e: \mathcal{A}^{l|In(tail(e))|} \rightarrow \mathcal{A}^l$. And for any terminal $t\in T$, there is a function $f_t: \mathcal{A}^{l|In(t))|} \rightarrow \mathcal{A}^r$.

In an $(r,l)$ fractional linear network code over a finite field $\mathbb{F}_q$, the above edge functions and decoding functions are linear functions. An $(r,l)$ fractional linear network code, in terms of a message transmitted over an edge and the message decoded by a terminal is described as follows. Let the message transmitted over an edge $e$ be denoted by $Y_e$. For any edge $e$ emanating from a source, $Y_e = A_{s_i,e}X_i$ where $Y_e\in \mathbb{F}_q^l, A_{s_i,e} \in \mathbb{F}_q^{l\times r}$, and $X_i\in \mathbb{F}_q^r$. For any edge emanating from an intermediate node $Y_e = A_{e_1,e}Y_{e_1} + A_{e_2,e}Y_{e_2} + \cdots + A_{e_n,e}Y_{e_n}$ where $In(tail(e)) = \{e_1,e_2,\ldots e_n\}$, $Y_e, Y_{e_1}, Y_{e_2}, \ldots , Y_{e_n} \in \mathbb{F}_q^l$ and $A_{e_1,e}, A_{e_2,e},\ldots, A_{e_n,e} \in \mathbb{F}_q^{l\times l}$. For any terminal $t_i\in T$,  $Z_{t_i} = A_{e_1,t_i}Y_{e_1} + A_{e_2,t_i}Y_{e_2} + \cdots + A_{e_n,t_i}Y_{e_n}$ where $In(t_i)= \{e_1,e_2,\ldots e_n\}$, $Z_{t_i}\in \mathbb{F}_q^r$, $Y_e, Y_{e_1}, Y_{e_2}, \ldots , Y_{e_n} \in \mathbb{F}_q^l$ and $A_{e_1,t_i}, A_{e_2,t_i},\ldots, A_{e_n,t_i} \in \mathbb{F}_q^{r\times l}$.

Let a block of $r$ symbols from a source $s_i$ be represented by $X_i=(X_{i1}, X_{i2}, \ldots, X_{ir}).$ A sum-network is said to have an $(r,l)$ fractional network coding solution over $\mathcal{A}$ if using an $(r,l)$ fractional network code all terminals can retrieve the sum $\sum_{i=1}^{|S|} X_i$ of the source processes. Analogously, a sum-network is said to have an $(r,l)$ fractional linear network coding solution over $\mathbb{F}_q$ if there exists an $(r,l)$ fractional linear network code over $\mathbb{F}_q$ such that all terminals can compute the sum of the source processes i.e. $\sum_{i=1}^{|S|} X_i$. The ratio $\frac{r}{l}$ is called the rate. A $(dr,dl)$ fractional (linear) network coding solution, where $d$ is any non-zero positive integer, is referred to as a rate $\frac{r}{l}$ achieving fractional (linear) network coding solution. $\frac{r}{l}$ is called a rate achievable through (linear) network coding if there exists a rate $\frac{r}{l}$ achieving fractional (linear) network coding solution. The capacity of a sum-network is defined as the supremum of all achievable rates. The linear coding capacity over a finite field $\mathbb{F}_q$ is the supremum of all rates achievable through linear network coding. The linear coding capacity is defined as supremum of the linear coding capacities over all finite fields. If a sum-network has a $(k,k)$ (for any positive integer $k\geq 1$) fractional (linear) network coding solution, then the sum-network is said to (linearly) solvable.

\section{A Sum-network (of capacity $k/n$) which has a $(k,n)$ fractional linear network coding solution only over a field of characteristic from a given finite set of primes}\label{sec2}
In this section we show that for any positive rational number $k/n$ and for any set of prime numbers $\{p_1,p_2,\ldots,p_w\}$, there exists a sum-network of capacity equal to $k/n$ such that the sum-network has a rate $\frac{k}{n}$ achieving fractional linear network coding solution if and only if the characteristic of the finite field belongs to the given set.

\begin{figure*}
\centering
\includegraphics[width=\textwidth]{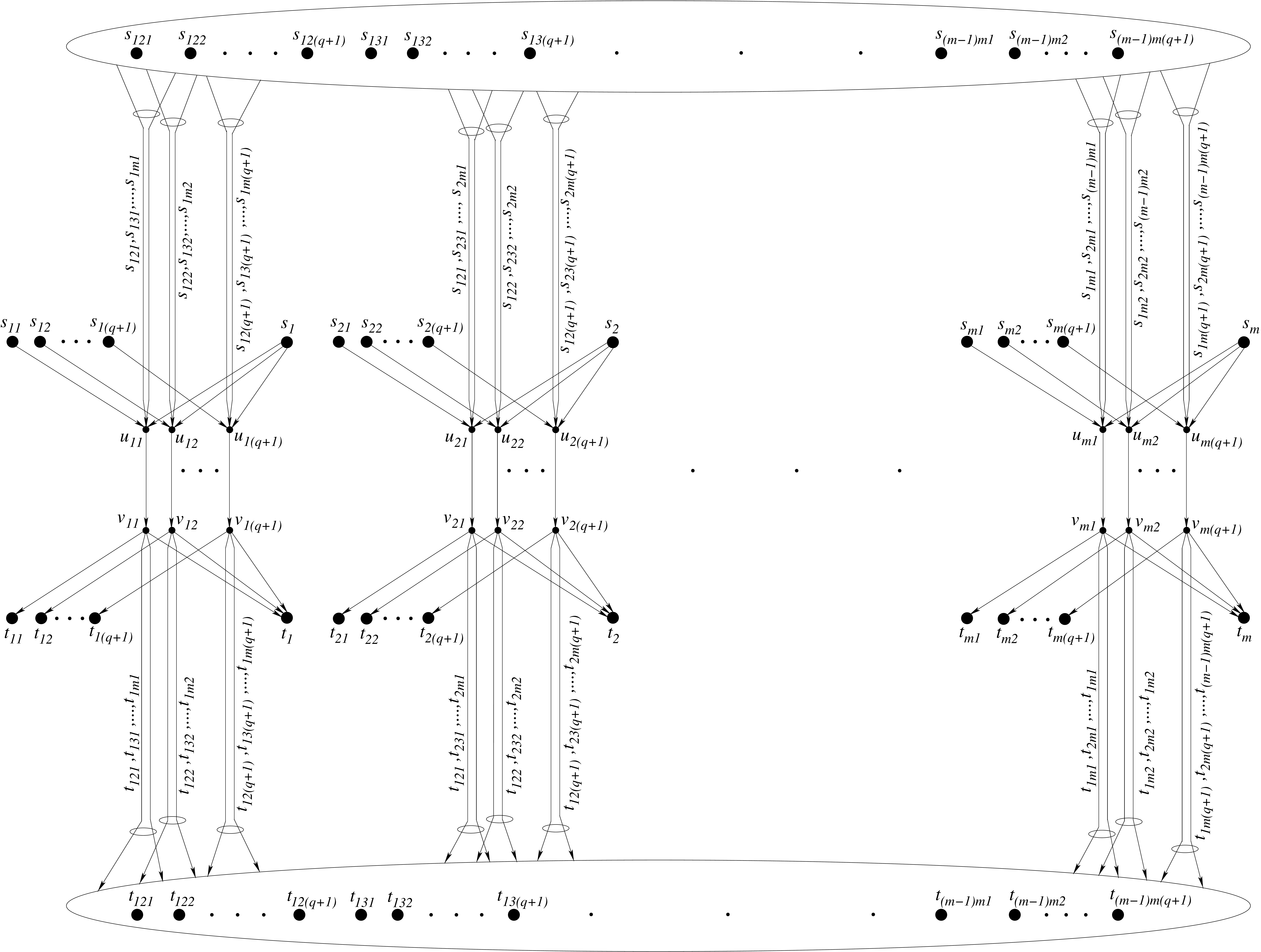}
\caption{A sum-network $\mathcal{N}_1$ whose capacity is $\frac{2}{m+1}$. $\mathcal{N}_1$ has a rate $\frac{2}{m+1}$ achieving fractional linear network coding solution if and only if the characteristic of the finite field divides $q$.}
\label{sumx}
\end{figure*}

Towards this end, we consider the sum-network $\mathcal{N}_1$ shown in Fig.~\ref{sumx}. 
$\mathcal{N}_1$ has $m + m(q+1) + \binom{m}{2}(q+1)$ many sources and the same number of terminals. The set of sources, $S$ is classified into three disjoint sets, namely $S_1,S_2$ and $S_3$; where, $S_1 = \{s_{i}|1\leq i\leq m\}$, $S_2 = \{ s_{ij} | 1\leq i\leq m, 1\leq j\leq (q+1) \}$ and $S_3 = \{ s_{ijx} | 1\leq i< m, i< j\leq m, 1\leq x\leq q+1 \}$. Similarly the set of terminals, $T$ is classified into $T_1,T_2$ and $T_3$, where $T_1 = \{t_{i}|1\leq i\leq m\}$, $T_2 = \{ t_{ij} | 1\leq i\leq m, 1\leq j\leq (q+1) \}$ and $T_3 = \{ t_{ijx} | 1\leq i < m, i< j\leq m, 1\leq x\leq q+1 \}$. Apart from the sources and terminals, there are $2(m)(q+1)$ intermediate nodes. These nodes are $\{u_{ij} | 1\leq i \leq m, 1\leq j\leq (q+1) \}$ and $\{v_{ij} | 1\leq i \leq m, 1\leq j\leq (q+1) \}$. The set of edges in the network are listed below.
\begin{enumerate}

\item $e_{ij} = (u_{ij},v_{ij})$ for $1\leq i\leq m$ and $1\leq j\leq (q+1)$.
\item $(s_i,u_{ij})$ for $1\leq i\leq m$ and $1\leq j\leq (q+1)$.
\item $(s_{ij},u_{ij})$ for $1\leq i\leq m$ and $1\leq j\leq (q+1)$.
\item $(s_{ijx},u_{ix})$ and $(s_{ijx},u_{jx})$ for $1\leq i < m, i< j\leq m$ and $1\leq x\leq (q+1)$.
\item $(v_{ij},t_i)$ for $1\leq i\leq m$ and $1\leq j\leq (q+1)$.
\item $(v_{ij},t_{ij})$ for $1\leq i\leq m$ and $1\leq j\leq (q+1)$.
\item $(v_{ix},t_{ijx})$ and $(v_{jx},t_{ijx})$ for $1\leq i < m, i< j\leq m$ and $1\leq x\leq (q+1)$.
\end{enumerate}
The edges $e_{ij}$ for $1\leq i\leq m$ and $1\leq j\leq (q+1)$ are shown as the middle edges in the network. It can be seen that the sources $s_i$ has an edge connecting to all of the middle edges $e_{ij}$ for $1\leq j\leq (q+1)$. Similarly the terminal $t_i$ has paths from all the the edge $\{e_{ij} | 1\leq j\leq (q+1) \}$. The source $s_{ij}$ has a connection to only one of the middle edges which is $e_{ij}$. Analogously the terminal $t_{ij}$ has an connection from only one middle edge $e_{ij}$. For any tuple $(i,x)$ where $1\leq i < m, i < x \leq m$, there exists $(q+1)$ sources in $S_3$. These sources are noted as $s_{ixj}$ where $1\leq j\leq (q+1)$. The source $s_{ixj}$ has paths to the middle edges $e_{ij}$ and $e_{xj}$. In the same way, the terminal $t_{ixj}$ has a connection from only two middle edges, which are $e_{ij}$ and $e_{xj}$. 

There are various direct edges between the sources and the terminals present in the network. The list of such edges are given below. Note that to reduce clumsiness, these direct edges have not been shown in Fig~\ref{sumx}. Say $S_{ij}$ denotes the set of sources which have a path to $tail(e_{ij})$ for $1\leq i\leq m$ and $1\leq j\leq (q+1)$. Here,
\begin{IEEEeqnarray*}{l}
S_{ij} = \{ s_i, s_{ij}\} \cup \{ s_{xij} | 1\leq x < i \} \cup \{ s_{ixj} | i < x \leq m \}
\end{IEEEeqnarray*}
Also let $S = S_1 \cup S_2 \cup S_3$.
\begin{enumerate}
\setcounter{enumi}{7}
\item $(s,t_i)$ for $1\leq i\leq m$ and $s \in S, s \notin \{ \cup_{j=1}^{q{+}1} S_{ij} \}$.
\item $(s,t_{ij})$ for $1\leq i\leq m, 1\leq j\leq (q+1)$ and $s \in S, s \notin S_{ij}$.
\item $(s,t_{ixj})$ for $1\leq i < m, i < x \leq m, 1\leq j\leq (q+1)$ and $s\in S, s \notin \{S_{ij}\cup S_{xj}\}$. 
\end{enumerate}

\begin{claim}\label{cla1}
The sum-network $\mathcal{N}_1$ shown in Fig.~\ref{sumx} has capacity equal to $\frac{2}{m+1}$; and has a rate $\frac{2}{m+1}$, fractional linear network coding solution if the characteristic of the finite field divides $q$.
\end{claim}
\begin{IEEEproof}
Consider an $(r,l)$ fractional linear network coding solution of the network. We use the notation $\{A\}\rightarrow \{B\}$ to indicate that the elements of the set $B$ can be computed from the elements of the set $A$. We note that a similar notation has been used earlier in \cite{doug}. The message transmitted by the edge $e_{ij}$ is denoted by $Y_{e_{ij}}$. For the sum-network to be solvable, the terminal $t_{iyj}$ for $1\leq i < m, i <y \leq m, 1\leq j\leq (q+1)$ must compute the following from $Y_{e_{ij}}$ and $Y_{e_{yj}}$:
\begin{IEEEeqnarray}{l}
X_i + X_y + X_{ij} + X_{yj} + \sum_{x=1}^{i-1} (X_{xij} + X_{xyj})\IEEEnonumber \\+\> \sum_{x=i+1}^{y-1} (X_{ixj} + X_{xyj}) + \sum_{x=y+1}^m (X_{ixj} + X_{yxj}) + X_{iyj}\label{eqn1}
\end{IEEEeqnarray}
However, from $Y_{e_{ij}}$, the terminal $t_{ij}$ computes 
\begin{equation}
X_i + X_{ij} + \sum_{x=1}^{i-1} X_{xij} + \sum_{x=i+1}^{y-1} X_{ixj} + X_{iyj} + \sum_{x=y+1}^{m} X_{ixj}\label{eqn2}
\end{equation}
and the terminal $t_{yj}$, from $Y_{e_{yj}}$, computes
\begin{equation}
X_y + X_{yj} + \sum_{x=1}^{i-1} X_{xyj} + X_{iyj} + \sum_{x=i+1}^{y-1} X_{xyj} + \sum_{x=y+1}^m X_{yxj}\label{eqn3}
\end{equation}
Now, first by adding equation (\ref{eqn2}) and equation (\ref{eqn3}), and then by subtracting equation (\ref{eqn1}) from the resultant it can be seen that the terminal $t_{iyj}$ can compute the source message $X_{iyj}$. Then we have,
\begin{IEEEeqnarray*}{l}
\{Y_{e_{ij}} | 1\leq i\leq m, 1\leq j\leq q{+}1\}\\
\rightarrow\\
\bigl\{\{X_i {+} X_{ij} {+} \sum_{x=1}^{i-1} X_{xij} {+} \sum_{x=i+1}^m X_{ixj} | {1{\leq} i{\leq} m,1\leq j\leq q{+}1}\}\\
\hfill \text{[because of the terminal $t_{ij}$]},\\
\{X_{ixj}| 1\leq i <m, i<x\leq m, 1\leq j\leq q{+}1\}  \\ 
\hfill \text{ [by equation~(\ref{eqn2}) + equation~(\ref{eqn3}) - equation~(\ref{eqn1})  ]}\bigr\}
\end{IEEEeqnarray*}
Hence,
\begin{IEEEeqnarray*}{l}
\{X_i | 1\leq i\leq m \}, \{Y_{e_{ij}} | 1\leq i\leq m, 1\leq j\leq q{+}1\}\\
\rightarrow\\
\bigl\{\{X_i | 1\leq i\leq m \}, \{X_{ij} | 1\leq i\leq m, 1\leq j\leq q{+}1 \},\\ \{X_{ixj}| 1\leq i <m, i<x\leq m, 1\leq j\leq q{+}1\}\bigr\}\IEEEyesnumber \label{mq1}
\end{IEEEeqnarray*}
Since, in a function, the set of possible values of the output of the function cannot be greater than the set of possible values that the argument can take, we must have,
\begin{IEEEeqnarray*}{l}
rm + lm(q+1) \geq rm + rm(q+1) +  \frac{rm(m-1)(q+1)}{2}\\
or,\, lm(q+1) \geq rm(q+1) +  \frac{rm(m-1)(q+1)}{2}\\
or,\, l \geq r +  \frac{r(m-1)}{2}\\
or,\, \frac{l}{r} \geq 1 + \frac{(m-1)}{2}\\
or,\, \frac{l}{r} \geq \frac{m+1}{2}\\
or,\, \frac{r}{l} \leq \frac{2}{m+1}
\end{IEEEeqnarray*}

We now show that $\mathcal{N}$ has a $(2,m+1)$ fractional linear network coding solution if the characteristic of the finite field divides $q$. Say the finite field is $\mathbb{F}_{p^a}$ where $p$ is a prime number, $p$ divides $q$, and $a$ is any non-zero positive integer. Since $r=2$, assume each source message is a $2$-length vector and each edge carry an $(m+1)$-length vector; where all elements of the vectors are from $\mathbb{F}_{p^a}$. Say the first and second symbols of $X_i, X_{ij}$ and $X_{yxj}$ are denoted by $X_{ia}, X_{ija},X_{yxja}$ and $X_{ib}, X_{ijb}, X_{yxjb}$ respectively where $1\leq i\leq m, 1\leq y < m, i < x \leq m, 1\leq j\leq (q{+}1)$.

In the first two time slots, all the edges $e_{ij}$ carries $ Y_{ij}^\prime =  X_i + X_{ij} + \sum_{x=1}^{i-1} X_{xij} + \sum_{x=i+1}^m X_{ixj}$  for $1\leq i\leq m$ and $1\leq j\leq (q+1)$; and all direct edges simple forwards the two symbols from the source at its head.
In the next $m-1$ time slots, edge $e_{ij}$ for $1\leq i\leq m, 1\leq j\leq (q+1)$ carries the $m-1$ symbols contained in the sets $\{X_{ixja}| i < x \leq m \}$ and $\{X_{xijb}| 1\leq x < i\}$. Note that since there is no source of the form $s_{iij}$ for $1\leq i\leq m$ and $1\leq j\leq (q+1)$, these two sets contains exactly $m-1$ elements. 

After the first two time slots, from the edge $e_{ij}$, for $1\leq i\leq m, 1\leq j\leq (q+1)$, terminal $t_i\in T_1$ receives, $Y_{ij}^\prime$. Summing this information for $1\leq j\leq (q+1)$ it gets:
\begin{IEEEeqnarray*}{l}
\sum_{j=1}^{q{+}1} Y_{ij}^\prime = (q+1)X_i + \sum_{j=1}^{q+1} ( X_{ij} + \sum_{x=1}^{x=i-1} X_{xij} +  \sum_{x=i+1}^m X_{ixj}
\end{IEEEeqnarray*}
Since $p$ divides the $q$, $q=0$ in $\mathbb{F}_{p^a}$ and hence $(q+1)=1$. Then, it can be seen that on summing the vectors coming from the direct edges with $\sum_{j=1}^{q{+}1} Y_{ij}^\prime$, terminal $t_i$ computes the sum.

Terminal $t_{ij}\in T_2$, from edge $e_{ij}$ for $1\leq i\leq m$ and $1\leq j\leq (q+1)$, after the first two time slots, receives $Y_{e_{ij}}^\prime$. On summing $Y_{e_{ij}}^\prime$ with the vectors coming in from the direct edges, it can be seen that terminal $t_{ij}$ can retrieve the required sum.

Now consider the terminal $t_{iyj}$ for $1\leq i < m, i< y\leq m$ and $1\leq j\leq q{+}1$. Note that $t_{iyj}$ is connected to the head nodes of the edges $e_{ij}$ and $e_{yj}$. Hence, after the first two time slots, it receives $Y_{e_{ij}}^\prime$ and $Y_{e_{yj}}^\prime = X_y + X_{yj} + \sum_{x=1}^{x=y-1} X_{xyj} + \sum_{x=y+1}^m X_{yxj}$. Since $y>i$ in $t_{iyj}$, in the rest of the $m-1$ time slots, from $e_{ij}$ it receives $X_{iyja}$, and from $e_{yj}$ it receives $X_{iyjb}$, and thus it gets the message $X_{iyj}$. Now it can be seen that by performing the operation $Y_{e_{ij}}^\prime + Y_{e_{yj}}^\prime - X_{iyj}$ terminal $t_{iyj}$ computes the sum shown in equation (\ref{eqn1}). On summing this information with the vectors coming from the direct edges it can be seen that the terminal $t_{iyj}$ can compute the required sum.
\end{IEEEproof}

\begin{claim}\label{cla2}
The linear coding capacity of the sum-network $\mathcal{N}_2$ shown in Fig.~\ref{sumx} is strictly less than $\frac{2}{m + 1}$ if the characteristic of the finite field does not divides $q$. 
\end{claim}
\begin{IEEEproof}
We have the following from the sum-network.
\begin{IEEEeqnarray*}{l}
\{Y_{e_{ij}} | 1\leq i\leq m, 1\leq j\leq q{+}1\}\\
\rightarrow\\
\bigl\{\{ X_i {+} \sum_{j=1}^{q{+}1} X_{ij} {+} \sum_{x=1}^{x=i-1}\sum_{j=1}^{q{+}1}\! X_{xij} {+} \sum_{x=i+1}^{m}\sum_{j=1}^{q{+}1}\! X_{ixj} | 1{\leq} i{\leq} m \}\\
\hfill \text{[because of the terminal $t_{i}$]},\\
\{X_i {+} X_{ij} {+} \sum_{x=1}^{i-1} X_{xij} {+} \sum_{x=i+1}^m X_{ixj} | {1\leq i\leq m,1\leq j\leq q{+}1}\}\\
\hfill \text{[because of the terminal $t_{ij}$]},\\
\{X_{ixj}| 1\leq i <m, i<x\leq m, 1\leq j\leq q{+}1\}  \\ 
\hfill \text{ [by equation~(\ref{eqn2}) + equation~(\ref{eqn3}) - equation~(\ref{eqn1})  ]}\bigr\}\\[12pt]
or, \{Y_{e_{ij}} | 1\leq i\leq m, 1\leq j\leq q{+}1\}\\
\rightarrow\\
\bigl\{\{ X_i {+} \sum_{j=1}^{q{+}1} X_{ij} | 1\leq i\leq m \},\\
\{X_i {+} X_{ij} | {1\leq i\leq m,1\leq j\leq q{+}1}\},\\
\{X_{ixj}| 1\leq i <m, i<x\leq m, 1\leq j\leq q{+}1\}\bigr\} 
\end{IEEEeqnarray*}
By performing the operation $\sum_{j=1}^{q{+}1} (X_i {+} X_{ij})$ for $1\leq i \leq m$, we get $(q{+}1)X_i + \sum_{j=1}^{q{+}1} X_{ij}$. Now, by noting that $(q{+}1)X_i + \sum_{j=1}^{q{+}1} X_{ij} - (X_i {+} \sum_{j=1}^{q{+}1} X_{ij}) = qX_i$, we get,
\begin{IEEEeqnarray*}{l}
\{Y_{e_{ij}} | 1\leq i\leq m, 1\leq j\leq q{+}1\}\\
\rightarrow\\
\bigl\{\{ qX_i | 1\leq i\leq m \},\\
\{X_i {+} X_{ij} | {1\leq i\leq m,1\leq j\leq q{+}1}\},\\
\{X_{ixj}| 1\leq i <m, i<x\leq m, 1\leq j\leq q{+}1\} \bigr\}
\end{IEEEeqnarray*}
Since, the characteristic of the finite field does not divide $q$, there exist an inverse of $q$. So,
\begin{IEEEeqnarray*}{l}
\{Y_{e_{ij}} | 1\leq i\leq m, 1\leq j\leq q{+}1\}\\
\rightarrow\\
\bigl\{\{ X_i | 1\leq i\leq m \},\\
\{X_i {+} X_{ij} | {1\leq i\leq m,1\leq j\leq q{+}1}\},\\
\{X_{ixj}| 1\leq i <m, i<x\leq m, 1\leq j\leq q{+}1\}\bigr\}  \\[12pt] 
or,\{Y_{e_{ij}} | 1\leq i\leq m, 1\leq j\leq q{+}1\}\\
\rightarrow\\
\bigl\{\{ X_i | 1\leq i\leq m \},\\
\{X_{ij} | {1\leq i\leq m,1\leq j\leq q{+}1}\},\\
\{X_{ixj}| 1\leq i <m, i<x\leq m, 1\leq j\leq q{+}1\}\bigr\}  
\end{IEEEeqnarray*}
\begin{IEEEeqnarray*}{l}
lm(q+1) \geq rm + rm(q+1) +  \frac{rm(m-1)(q+1)}{2}\\
or,\, l(q+1) \geq r + r(q+1) +  \frac{r(m-1)(q+1)}{2}\\
or,\, l(q+1) \geq r( 1 + (q+1) +  \frac{(m-1)(q+1)}{2})\\
or,\, \frac{r}{l} \leq \frac{q+1}{1 + (q+1) +  \frac{(m-1)(q+1)}{2}}\\
or,\, \frac{r}{l} \leq \frac{1}{\frac{1}{q+1} + 1 +  \frac{(m-1)}{2}}\\
or,\, \frac{r}{l} \leq \frac{2}{\frac{2}{q+1} + 2 +  (m-1)}\\
or,\, \frac{r}{l} \leq \frac{2}{\frac{2}{q+1} + 1 + m}
\end{IEEEeqnarray*}
Since, $\frac{2}{q+1} > 0$, it follows that $\frac{2}{\frac{2}{q+1} + 1 + m} <\frac{2}{1 + m}$.
\end{IEEEproof}

\begin{theorem}\label{thm3}
For any non-zero positive rational number $\frac{k}{n}$, and for any given set of primes $\{p_1,p_2,\ldots,p_w\}$,
there exist a sum-network which has capacity equal to $\frac{k}{n}$, and has a capacity achieving rate $\frac{k}{n}$ fractional linear network coding solution if and only if the characteristic of the finite field belongs to the given set.
\end{theorem}
\begin{IEEEproof}
Consider the sum-network $\mathcal{N}_1$ shown in Fig.~\ref{sumx} with $q = p_1\times p_2\times\cdots\times p_w$ and $m=2n-1$. Now make $k$ disjoint copies of $\mathcal{N}_1$ and name the $i^{\text{th}}$ copy as $\mathcal{N}_{1i}$. Also let us name the edge $e_{ij}$ in $\mathcal{N}_{1x}$ as $e_{ijx}$ where $1\leq i\leq m, 1\leq j\leq q+1$ and $1\leq x\leq k$. Note that, there are now $k$ copies of any source $s_i$, and $k$ copies of any terminal $t_i$. Combine all sources and terminals designated by the same label into a single source and a single terminal respectively;	 \textit{i.e.} all $k$ copies of the source $s_i$ are combined into one source denoted by the same label $s_i$; and all $k$ copies of the terminal $t_i$ are combined into one terminal and denoted by $t_i$. Let us name the resulted sum-network as $\mathcal{N}_1^\prime$.

Now, going through a proof similar to that of Claim~\ref{cla1}, it can be seen that instead if equation $(\ref{mq1})$ the following equation holds true:
\begin{IEEEeqnarray*}{l}
\bigl\{\{X_i | 1{\leq} i{\leq} m \}, \{Y_{e_{ijx}} | 1{\leq} i{\leq} m, 1\leq j\leq q{+}1, 1\leq x\leq k\}\bigr\}\\
\rightarrow\\
\bigl\{\{X_i | 1\leq i\leq m \}, \{X_{ij} | 1\leq i\leq m, 1\leq j\leq q{+}1 \},\\ \{X_{ixj}| 1\leq i <m, i<x\leq m, 1\leq j\leq q{+}1\}\bigr\} 
\end{IEEEeqnarray*}
From this equation, proceeding as in Claim~\ref{cla1}, it can be seen that
\begin{equation}
\frac{r}{l} \leq \frac{2k}{1 + m}
\end{equation}
Hence, the sum-network has capacity equal to $\frac{k}{n}$ when $n=2m-1$. Now we show that the sum-network has a $(2k,1+m)$ fractional linear network coding solution if the characteristic of the finite field divides $q$. Note that, since $q = p_1\times p_2\times\cdots\times p_w$, the characteristic divides $q$ if and only if it belongs to the set $\{p_1,p_2,\ldots,p_w\}$. Say the symbol generated at the source $s_i$ is represented as $X_i = (X_{i1}, X_{i2},\ldots,X_{i(2k)}$. Also let us represent the sum $\sum_{s_i\in S} X_i$ by $X$. Hence, $X = (\sum_{s_i\in S} X_{i1}, \sum_{s_i\in S} X_{i2}, \ldots , \sum_{s_i\in S} X_{i(2k)})$ where $S$ is the set of all sources. As shown in Claim~\ref{cla1}, since the sum-network $\mathcal{N}_{1x}$ has a $(2,m+1)$ fractional linear network coding solution if the characteristic divides $q$, the terminal $t_i$ can receive the information $\sum_{s_j\in S} X_{j(2x)}$ and $\sum_{s_j\in S} X_{j(2x-1)}$ from the $x^{\text{th}}$ copy \textit{i.e.} $\mathcal{N}_{1x}$ where $1\leq x\leq k$. This forms the required fractional linear network coding solution.

We now show that the linear coding capacity of $\mathcal{N}_1^\prime$ is strictly less than $\frac{2k}{1 + m}$ if the characteristic of the finite field does not divides $q$. Note that the characteristic of the finite field does not divides $q$ if and only if the characteristic does not belong to the set $\{p_1,p_2,\ldots,p_w\}$. Consider the following lemma which is a reproduction of Lemma~6 given in \cite{das2}, but here for a different network.
\begin{lemma}\label{lem1}
If the sum-network $\mathcal{N}_1^\prime$ has an $(r,l)$ fractional linear network coding solution then the sum-network $\mathcal{N}_1$ has an $(r,lk)$ fractional linear network coding solution.
\end{lemma}
\begin{proof}
For any edge $e$ in $\mathcal{N}_1$ there exists $k$ copies of $e$ in $\mathcal{N}_1^\prime$. So the information that can be transmitted though the $k$ copies of $e$ in one usage of $\mathcal{N}_1^\prime$, can be transmitted though $e$ in $k$ usage of $\mathcal{N}_1$. So by using the sum-network $\mathcal{N}_1$ for $k$ times, the $(r,lk)$ fractional linear network coding solution can be achieved.
\end{proof}
Let us assume that $\mathcal{N}_1^\prime$ has a rate $\frac{2k}{1+m}$ achieving fractional linear network coding solution when the characteristic of the finite field does not divide $q$. In which case as per Lemma~\ref{lem1}, the sum-network has a rate $\frac{2k}{k(1+m)} = \frac{2}{1+m}$ achieving fractional linear network coding solution. However this is in contradiction with Claim~\ref{cla2}. Hence the sum-network $\mathcal{N}_1^\prime$ is the required sum-network.
\end{IEEEproof}

\section{A Sum-network (of capacity $k/n$) which has a $(k,n)$ fractional linear network coding solution only over a field of characteristic from a given co-finite set of primes}\label{sec3}

\begin{figure*}
\centering
\includegraphics[width=\textwidth]{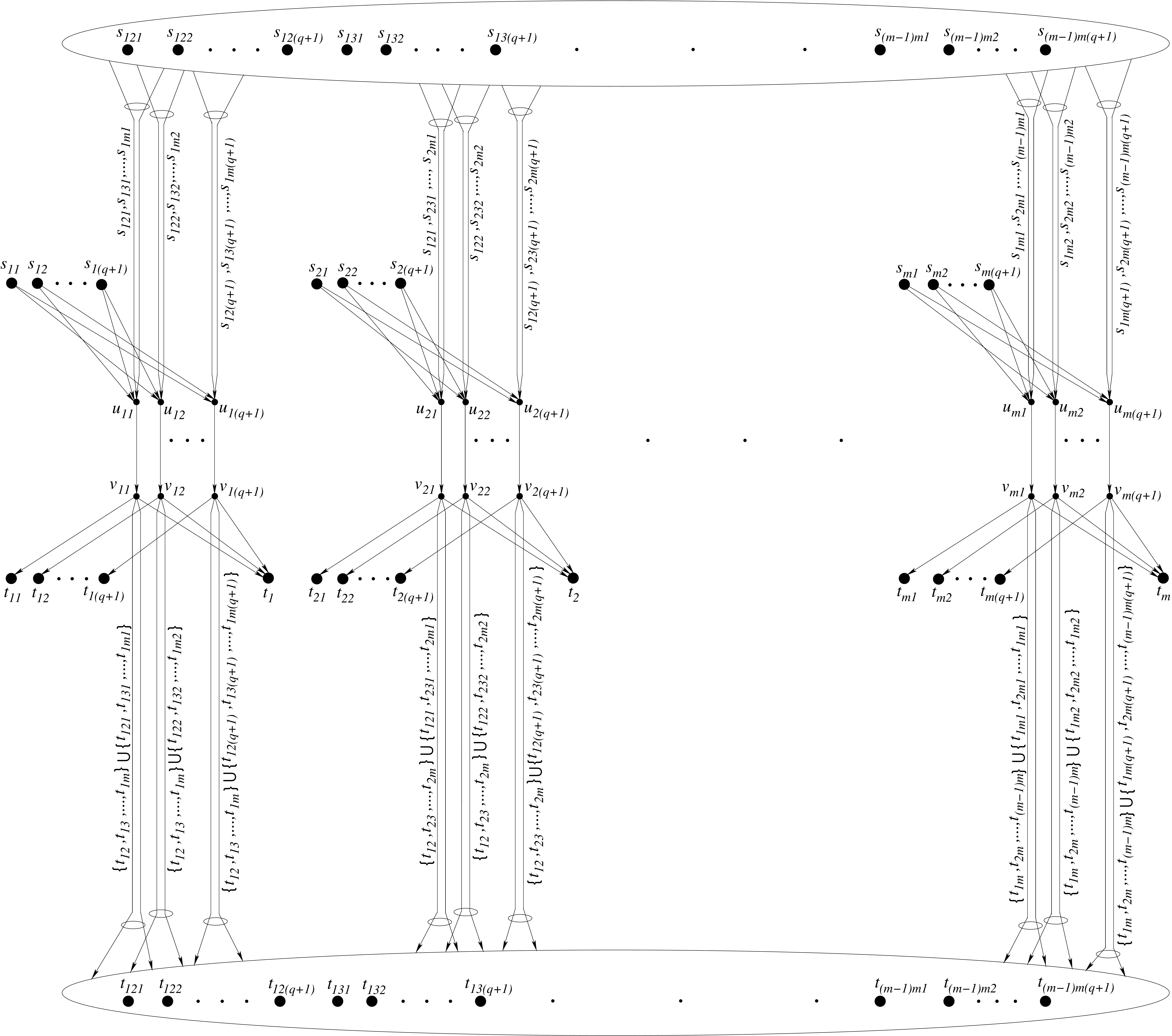}
\caption{A sum-network $\mathcal{N}_2$ whose capacity is $\frac{2}{m+1}$. $\mathcal{N}_2$ has a rate $\frac{2}{m+1}$ achieving fractional linear network coding solution if and only if the characteristic of the finite field does not divide $q$.}
\label{nonsum}
\end{figure*}

In this section we show that for any positive rational number $k/n$ and any set of prime number $\{p_1,p_2,\ldots,p_w\}$ there exists a sum-network of capacity $\frac{k}{n}$, which has a rate $\frac{k}{n}$ achieving fractional linear network coding solution if and only if the characteristic of the finite field does not belong to the given set of primes. 

For this purpose, we consider the sum-network $\mathcal{N}_2$ shown in Fig.~\ref{nonsum}. The sum-network has $m(q+1) + \binom{m}{2}(q+1)$ sources and $m + m(q+1) + \binom{m}{2}(q+1) + \binom{m}{2}$ many terminals. The set of sources are partitioned into two sets: $S_1 = \{s_{ij} | 1\leq i\leq m, 1\leq j \leq (q+1) \}$ and $S_2 = \{ s_{ixj} | 1\leq i < m, i < x \leq m, 1\leq j\leq (q+1) \}$. The set of terminals are partitioned into four sets: $T_1 = \{t_i | 1\leq i\leq m \}$, $T_2 = \{t_{ij} | 1\leq i\leq m, 1\leq j \leq (q+1) \}$, $T_3 = \{ t_{ijx} | 1\leq i < m, i < j \leq m, 1\leq x\leq (q+1) \}$ and $T_4 = \{t^\prime_{ij} | 1\leq i,j\leq m, i < j \}$. Apart from the sources and terminals, there are $2(m)(q+1)$ intermediate nodes. These nodes are $\{u_{ij} | 1\leq i \leq m, 1\leq j\leq (q+1) \}$ and $\{v_{ij} | 1\leq i \leq m, 1\leq j\leq (q+1) \}$. The set of edges in the sum-network are as follows.
\begin{enumerate}
\item $e_{ij} = (u_{ij},v_{ij})$ for $1\leq i\leq m$ and $1\leq j\leq (q+1)$.
\item $(s_{ij},u_{ix})$ for $1\leq i\leq m$, $1\leq j,x\leq (q+1)$ and $x\neq j$.
\item $(s_{ijx},u_{iy})$ and $(s_{ijx},u_{jy})$ for $1\leq i < m, i< j\leq m$, $1\leq x,y\leq (q+1)$ and $y\neq x$.
\item $(v_{ij},t_i)$ for $1\leq i\leq m$ and $1\leq j\leq (q+1)$.
\item $(v_{ij},t_{ij})$ for $1\leq i\leq m$ and $1\leq j\leq (q+1)$.
\item $(v_{ix},t_{ijx})$ and $(v_{jx},t_{ijx})$ for $1\leq i < m, i< j\leq m$ and $1\leq x\leq (q+1)$.
\item $(v_{ij},t^\prime_{xi})$ and $(v_{ij},t^\prime_{iy})$ for $1\leq i,x,y\leq m, x < i, y > i$ and $1\leq j\leq (q+1)$.
\end{enumerate}
The middle edges in the network are $e_{ij}$ for $1\leq i\leq m$ and $1\leq j\leq (q+1)$. As shown, the source $s_{ij}$ has an edge connected to each of the middle edges $e_{ik}$ for $1\leq k\leq (q+1)$ except the edge $e_{ij}$. And the source $s_{ixj}$ for $1\leq i < m, i < x \leq m$ and $1\leq j\leq (q+1)$ are connected to all of the middle edges in the set $\{e_{by} | b \in \{ i,x\}, 1\leq y\leq (q+1) \}$ except the edges $e_{ij}$ and $e_{xj}$. The terminal $t_{i}$ is connected to all of the middle edges $\{e_{ij}|1\leq j\leq (q+1)\}$. For $1\leq i\leq m$ and $1\leq j\leq (q+1)$, the terminal $t_{ij}$ has an edge from the middle edge $e_{ij}$. The terminal $t_{ixj}$ has edges from the head nodes of $e_{ij}$ and $e_{xj}$. For, $1\leq i,x\leq m$ and $i<x$, terminal $t_{ix}^\prime$ has an edge from all of the middle edges in the set $\{e_{yj} | y \in \{i,x\}, 1\leq j\leq (q+1)\}$. 

Let the set of sources which has a path to $tail(e_{ij})$ for $1\leq i\leq m$ and $1\leq j\leq (q+1)$ is denoted by $S_{ij}$. It can be seen that for $1\leq i\leq m$ and $1\leq j\leq q+1$,
\begin{IEEEeqnarray*}{l}
S_{ij} = \{ s_{ix} | 1\leq x \leq (q+1), x\neq j\}\\ \hfill \cup\> \{ s_{xiy} | 1\leq x < i, 1\leq y \leq (q+1), y\neq j \}\\ \hfill \cup\> \{ s_{ixy} | i < x \leq m, 1\leq y\leq (q+1), y\neq j \}
\end{IEEEeqnarray*}
Also let $S = S_1 \cup S_2$. We now list the direct edges. Note that these direct edges have not been shown in Fig.~\ref{nonsum}.
\begin{enumerate}
\setcounter{enumi}{7}
\item $(s,t_i)$ for $1\leq i\leq m$, $s\in S$ and $s \notin \cup_{j=1}^{q{+}1} S_{ij}$.
\item $(s,t_{ij})$ for $1\leq i\leq m$, $1\leq j\leq q+1$, $s\in S$ and $s \notin S_{ij}$.
\item $(s,t_{ijx})$ for $1\leq i,j\leq m$, $i < j$, $1\leq x\leq q+1$, $s\in S$ and $s \notin \{ S_{ix} \cup S_{jx}\}$.
\item $(s,t^\prime_{ij})$ for $1\leq i,j \leq m$, $i < j$, $s\in S$, $s \notin \cup_{x=1}^{q{+}1} S_{ix}$  and $s \notin \cup_{x=1}^{q{+}1} S_{jx}$. 
\end{enumerate}

\begin{claim}\label{cla3}
The sum-network in Fig.~\ref{nonsum} has capacity equal to $\frac{2}{m+1}$; and has a rate $\frac{2}{m+1}$ fractional linear network coding solution if the characteristic of the finite field does not divide $q$.
\end{claim}
\begin{IEEEproof}
Consider an $(r,l)$ fractional linear network coding solution of the sum-network. We first list which terminal has to compute which information from which set of edges for the sum-network to have such a solution.
\begin{enumerate}
\item For $1\leq i\leq m$, from the edge set $\{e_{ij} | 1\leq j \leq (q+1) \}$ the terminal $t_{i}$ must compute the sum:
\begin{IEEEeqnarray}{l}
Y_i = \sum_{x=1}^{q+1} X_{ix} + \sum_{x=1}^{i-1} \sum_{y=1}^{q+1} X_{xiy} + \sum_{x=i+1}^{m} \sum_{y=1}^{q+1} X_{ixy} \label{dq1}
\end{IEEEeqnarray}
\item For $1\leq i\leq m$ and $1\leq j\leq (q+1)$, from the edge $e_{ij}$ the terminal $t_{ij}$ must compute:
\begin{IEEEeqnarray}{l}
Y_{ij}^\prime = \sum_{x{=}1,x{\neq} j}^{q+1} X_{ix} {+} \sum_{x{=}1}^{i{-}1}\, \sum_{y{=}1,y{\neq}j}^{q{+}1} X_{xiy} \IEEEnonumber\\ \hfill +\> \sum_{x=i{+}1}^{m} \, \sum_{y{=}1,y{\neq} j}^{q+1} X_{ixy}\IEEEeqnarraynumspace \label{dq2}
\end{IEEEeqnarray}
\item For $1\leq i,x \leq m, i < x, 1\leq j\leq (q{+}1)$, from the two edges $e_{ij}$ and $e_{xj}$, the terminal $t_{ixj}$ must compute:
\begin{IEEEeqnarray}{l}
Y_{ixj} = \sum_{y=1,y\neq j}^{q+1} X_{iy} + \sum_{y=1,y\neq j}^{q+1} X_{xy} + \sum_{b=1}^{i-1} \sum_{y{=}1,y{\neq}j}^{q+1} X_{biy}\IEEEnonumber \\+\> \sum_{b=i+1,b\neq x}^{m} \sum_{y{=}1,y{\neq} j}^{q+1} X_{iby}  + \sum_{b=1,b\neq i}^{x-1} \sum_{y{=}1,y{\neq}j}^{q+1} X_{bxy}\IEEEnonumber \\+\> \sum_{b=x+1}^{m} \sum_{y{=}1,y{\neq} j}^{q+1} X_{xby} + \sum_{y{=}1,y{\neq}j}^{q+1} X_{ixy} \label{dq3}
\end{IEEEeqnarray}
\item For $1\leq i,j \leq m, i<j$, From the edges in the set $\{e_{ix} | 1\leq x\leq (q+1)\} \cup \{e_{jx} | 1\leq x\leq (q+1)\}$, the terminal $t^\prime_{ij}$ must compute:
\begin{IEEEeqnarray}{l}
Z_{ij} = \sum_{y=1}^{q+1} X_{iy} + \sum_{y=1}^{q+1} X_{jy} + \sum_{x=1}^{i-1} \sum_{y=1}^{q+1} X_{xiy} \IEEEnonumber\\+\> \sum_{x=i+1,x\neq j}^{m} \sum_{y=1}^{q+1} X_{ixy} \IEEEnonumber + \sum_{x=1,x\neq i}^{j-1} \sum_{y=1}^{q+1} X_{xjy} \\+\> \sum_{x=j+1}^{m} \sum_{y=1}^{q+1} X_{jxy} + \sum_{y=1}^{q+1} X_{ijy} \IEEEeqnarraynumspace\label{dq4}
\end{IEEEeqnarray}
\end{enumerate}
Now consider the following information that the terminals can also compute:
\begin{enumerate}
\item For $1\leq i,x\leq m, i < x, 1\leq j\leq (q+1)$, since $t_{ixj}$ is connected to edges $e_{ij}$ and $e_{xj}$ by the edges $(v_{ij},t_{ixj})$ and $(v_{xj},t_{ixj})$ respectively, according to equations (\ref{dq2}) and (\ref{dq3}), by doing the operation $Y_{ij}^\prime + Y_{xj}^\prime - Y_{ixj}$ terminal $t_{ixj}$ can compute:

\begin{equation}
\sum_{y=1,y\neq j}^{q+1} X_{ixy} \label{dq5}
\end{equation}
\item For $1\leq i,x\leq m, i<x$, since the information computed by the the terminals $t_i$ and $t_x$ can also be computed by the terminal $t_{ix}^\prime$, from equations (\ref{dq1}) and (\ref{dq4}), by doing the operation $Y_i + Y_x - Z_{ix}$ the terminal $t_{ix}^\prime$ can compute:
\begin{IEEEeqnarray}{l}
\sum_{y=1}^{q+1} X_{ixy} \label{dq6}
\end{IEEEeqnarray}
\end{enumerate}
Then, for $1\leq i,x\leq m$ and $i < x$, for each value of $j$ in the range $1$ to $(q+1)$, by doing the operation
\begin{equation*}
\sum_{y=1}^{q+1} X_{ixy} - \sum_{y=1,y\neq j}^{q+1} X_{ixy} \qquad \text{ (from eq. (\ref{dq5}) and eq. (\ref{dq6}))}
\end{equation*}
the terminal $t_{ix}^\prime$ computes the information $X_{ixj}$.

Hence we have,
\begin{IEEEeqnarray*}{l}
\{Y_{e_{ij}} | 1\leq i\leq m, 1\leq j\leq q{+}1\}\\
\rightarrow\\
\bigr\{\{ X_{ixy} | 1\leq i,x\leq m, i < x, 1\leq y\leq (q+1) \},\\
\{\sum_{x=1}^{q+1} X_{ix} + \sum_{x=1}^{i-1} \sum_{y=1}^{q+1} X_{xiy} + \sum_{x=i+1}^{m} \sum_{y=1}^{q+1} X_{ixy} | 1\leq i\leq m\}\\\hfill \text{ (from eq. } (\ref{dq1})),\\
\{\sum_{x=1,x\neq j}^{q+1} X_{ix} + \sum_{x=1}^{i-1} \sum_{y{=}1,y{\neq}j}^{q+1} X_{xiy} \\\hfill +\>  \sum_{x=i+1}^{m} \sum_{y{=}1,y{\neq} j}^{q+1} X_{ixy} | 1\leq i\leq m, 1\leq j\leq (q+1) \} \\\hfill \text{ (from eq. } (\ref{dq2}))
\bigr\}\\[12pt]
or, \{Y_{e_{ij}} | 1\leq i\leq m, 1\leq j\leq q{+}1\}\\
\rightarrow\\
\bigl\{\{ X_{ixy} | 1\leq i,x\leq m, i < x, 1\leq y\leq (q+1) \},\\
\{ \sum_{j=1}^{q+1} X_{ij} | 1\leq i\leq m \}, \IEEEyesnumber \label{wq1}\\
\{ \sum_{y=1,y\neq j}^{q+1} X_{iy} | 1\leq i\leq m, 1\leq j\leq (q+1) \}\bigr\}  \IEEEyesnumber \label{wq2}
\\[12pt]
or, \{Y_{e_{ij}} | 1\leq i\leq m, 1\leq j\leq q{+}1\}\\
\rightarrow\\
\bigl\{\{ X_{ixy} | 1\leq i,x\leq m, i < x, 1\leq y\leq (q+1) \},\\
\{ X_{ij} | 1\leq i\leq m, 1\leq j\leq (q+1) \}\bigr\}  \IEEEyesnumber\label{wq3}
\end{IEEEeqnarray*}
Hence, we have,
\begin{IEEEeqnarray*}{l}
lm(q+1) \geq \frac{rm(m-1)(q+1)}{2} + rm(q+1)\\
or,\, l \geq \frac{r(m-1)}{2} + r\\
or,\, \frac{l}{r} \geq \frac{(m-1)}{2} + 1\\
or,\, \frac{l}{r} \geq \frac{m+1}{2}\\
or,\, \frac{r}{l} \leq \frac{2}{m+1}
\end{IEEEeqnarray*}

We now show that $\mathcal{N}_2$ has a $(2,m+1)$ fractional linear network coding solution if the characteristic of the finite field does not divide $q$. As before, the two symbols generated at the sources are differentiated by indicating the first with an added suffix $a$, and indicating the second with an added suffix $b$. In the first two time slots, all the edges $e_{ij}$ carries, for $1\leq i\leq m$ and $1\leq j\leq (q+1)$,
\begin{IEEEeqnarray*}{l}
Y_{ij}^\prime = \sum_{x=1,x\neq j}^{q+1} X_{ix} + \sum_{x=1}^{i-1} \sum_{y{=}1,y{\neq}j}^{q+1} X_{xiy} + \sum_{x=i+1}^{m} \sum_{y{=}1,y{\neq} j}^{q+1} X_{ixy}
\end{IEEEeqnarray*}
And all direct edges simply forwards the two symbols from the source it is connected to. 
In the next $m-1$ time slots, edge $e_{ij}$ for $1\leq i\leq m, 1\leq j\leq (q+1)$ carries the $m-1$ symbols contained in sets $\{ \sum_{y{=}1,y{\neq}j}^{(q+1)} X_{ixya}| i < x \leq m \}$ and $\{ \sum_{y{=}1,y{\neq}j}^{(q+1)} X_{xiyb}| 1\leq x < i\}$.

Terminal $t_{ij}\in T_2$, from edge $e_{ij}$ for $1\leq i\leq m$ and $1\leq j\leq (q+1)$, after the first two time slots, receives $Y_{e_{ij}}^\prime$ which is same as equation (\ref{dq2}). So, it can compute the sum.

Terminal $t_i$ for $1\leq i\leq m$ adds $\sum_{j=1}^{q+1} Y_{ij}^\prime$ and gets:
\begin{IEEEeqnarray*}{l}
\sum_{j=1}^{q+1} Y_{ij}^\prime = \sum_{j=1}^{q+1} qX_{ij} + \sum_{x=1}^{i-1} \sum_{j=1}^{q+1} qX_{xij} +  \sum_{x=i+1}^m \sum_{j=1}^{q+1} qX_{ixj} 
\end{IEEEeqnarray*}
Now, since the characteristic of the finite field does not divides $q$, $q$ has an inverse, and hence terminal $t_i$ can compute $Y_i$ shown in equation (\ref{dq1}).

Terminal $t_{ixj}$ for $1\leq i,x\leq m, i < x, 1\leq j\leq q+1$, receives $\sum_{y=1,y\neq j}^{(q+1)} X_{ixya}$ and $\sum_{y=1,y\neq j}^{(q+1)}X_{ixjb}$ from the edges $e_{ij}$ and $e_{xj}$ respectively. Hence, by concatenation it computes the 2-length vector $\sum_{y=1,y\neq j}^{(q+1)} X_{ixy}$. Now by doing $Y_{ij}^\prime + Y_{xj}^\prime - \sum_{y=1,y\neq j}^{(q+1)} X_{ixj}$, terminal $t_{ixj}$ computes $Y_{ixj}$ as shown in equation (\ref{dq3}).

Now consider the terminal $t_{ix}^\prime$. As shown in the preceding paragraph, for $1\leq j\leq (q+1)$ the terminal $t_{ixj}$ can compute the sum $W_{ixj} = \sum_{y=1,y\neq j}^{(q+1)} X_{ixy}$. Then by doing the operation $\sum_{j=1}^{q+1} W_{ixj}$, the terminal $t_{ix}^\prime$ can compute,
\begin{IEEEeqnarray*}{l}
\sum_{j=1}^{q+1} W_{ixj} = \sum_{y=1}^{q+1} qX_{ixy} 
\end{IEEEeqnarray*}
Since $q$ has an inverse if the characteristic of the finite field does not divides $q$, terminal $t_{ixj}$ can compute the sum $\sum_{y=1}^{q+1} X_{ixy}$. Now, as already shown, terminal $t_i$ can compute $Y_i$ from $\{e_{ij} | 1\leq j\leq (q+1) \}$, and $t_x$ can compute $Y_x$ from $\{e_{xj} | 1\leq j\leq (q+1) \}$. Since the tail node of all of these edges are also connected to $t_{ix}^\prime$, by doing the operation $Y_i + Y_x - \sum_{y=1}^{q+1} X_{ixy}$, it can be seen that it computes the information $Z_{ix}$ shown in equation (\ref{dq4}).
\end{IEEEproof}

\begin{claim}\label{cla4}
The linear coding capacity of the sum-network $\mathcal{N}_2$ shown in Fig.~\ref{nonsum} is strictly less than $\frac{2}{m + 1}$ if the characteristic of the finite field divides $q$. 
\end{claim}
\begin{IEEEproof}
Let us consider the following notations.
\begin{IEEEeqnarray*}{l}
A = \{ Y_{e_{ij}} | 1\leq i\leq m, 1\leq j\leq q\}\\
B = \{ Y_{e_{i(q+1)}} | 1\leq i\leq m\}\\
C = \{ Y_{e_{i(q+1)}}^\prime | 1\leq i\leq m \}
\end{IEEEeqnarray*}
First we show that $H(A) = H(A,C)$ if the characteristic of the finite field divides $q$ in the following way. Note that for $1\leq i\leq m, 1\leq j\leq q$, from the edge $e_{ij}$, as shown in equation (\ref{dq2}), the information $Y_{ij}^\prime$ can be determined. Then, for $1\leq i\leq m$,
\begin{IEEEeqnarray*}{l}
\sum_{j=1}^{q} {Y_{ij}^\prime} =  \sum_{y=1}^q (q-1)X_{iy} + \sum_{x=1}^{i-1} \sum_{y=1}^q (q-1)X_{xiy} \\+\> \sum_{x=i+1}^{m} \sum_{y=1}^q (q-1)X_{ixy} + qX_{i(q+1)} + \sum_{x=1}^{i-1} qX_{xi(q+1)} \\+\> \sum_{x=i+1}^{m} qX_{xi(q+1)}\\
\text{Now, since $q = 0$ if the characteristic divides $q$, we have,}\\
\sum_{j=1}^{q} {Y_{ij}^\prime} =  \sum_{y=1}^q (q-1)X_{iy} + \sum_{x=1}^{i-1} \sum_{y=1}^q (q-1)X_{xiy} \\+\> \sum_{x=i+1}^{m} \sum_{y=1}^q (q-1)X_{ixy} \IEEEyesnumber \label{dq8}
\end{IEEEeqnarray*}
Since the characteristic cannot divide two consecutive elements of a field, there must exists an inverse of $(q-1)$; which is $(q-1)$ itself. Hence from equation (\ref{dq8}) we get,
\begin{IEEEeqnarray*}{l}
\sum_{j=1}^{q} (q-1){Y_{ij}^\prime} =  \sum_{y=1}^q X_{iy} + \sum_{x=1}^{i-1} \sum_{y=1}^q X_{xiy} \\+\> \sum_{x=i+1}^{m} \sum_{y=1}^q X_{ixy}\\
\text{However, from equation (\ref{dq2}) it can be seen that}\\
\sum_{j=1}^{q} (q-1){Y_{ij}^\prime} = Y_{i(q+1)}^\prime
\end{IEEEeqnarray*}
Hence $Y_{i(q+1)}^\prime$ for $1\leq i \leq m$ can be computed from $\{ Y_{e_{ij}} | 1\leq i\leq m, 1\leq j\leq q\}$ if the characteristic of the finite field divides $q$. This shows that 
\begin{equation}
H(A) = H(A,C) \label{dq9}
\end{equation}
Now, since equation (\ref{wq3}) is true irrespective of the characteristic of the finite field, we get:
\begin{IEEEeqnarray}{l}
H(A,B) \geq \sum_{i=1}^{m-1} \sum_{x=i+1}^{m} \sum_{y=1}^{q+1} H(X_{ixy}) + \sum_{i=1}^m \sum_{j=1}^{q+1} H(X_{ij})\IEEEeqnarraynumspace \label{dq7}
\end{IEEEeqnarray}
Now we have,
\begin{IEEEeqnarray*}{l}
\sum_{i=1}^m \sum_{j=1}^q H(Y_{e_{ij}}) + \sum_{i=1}^m H(Y_{e_{i(q+1)}}) - \sum_{i=1}^m H(Y_{i(q+1)}^\prime)\\
=\sum_{i=1}^m \sum_{j=1}^q H(Y_{e_{ij}}) + \sum_{i=1}^m H(Y_{e_{i(q+1)}},Y_{i(q+1)}) \\\hfill -\> \sum_{i=1}^m H(Y_{i(q+1)}^\prime) \, \text{ (from (\ref{dq2}))}\\
\geq H(A) + H(B,C) - H(C) \\
= H(A) + H(B|C)\\
\geq H(A) + H(B|C,A) \\
= H(A) + H(B|A) \qquad \text{ (from (\ref{dq9}))}\\
= H(A,B)\\
\geq \sum_{i=1}^{m-1} \sum_{x=i+1}^{m} \sum_{y=1}^{q+1} H(X_{ixy}) + \sum_{i=1}^m \sum_{j=1}^{q+1} H(X_{ij}) \hfill \text{ (from (\ref{dq7}))}
\end{IEEEeqnarray*}
Hence, we have,
\begin{IEEEeqnarray*}{l}
lm(q) + lm - rm \geq \frac{rm(m-1)(q+1)}{2} + rm(q+1)\\
lm(q+1) \geq \frac{rm(m-1)(q+1)}{2} + rm(q+1) + rm\\
or,\, l \geq \frac{r(m-1)}{2} + r + \frac{r}{q+1}\\
or,\, \frac{l}{r} \geq \frac{(m-1)}{2} + 1 + \frac{1}{q+1}\\
or,\, \frac{l}{r} \geq \frac{m+1}{2} + \frac{1}{q+1}\\
or,\, \frac{l}{r} \geq \frac{(m+1)(q+1) + 2}{2(q+1)}\\
or,\, \frac{r}{l} \leq \frac{2(q+1)}{(m+1)(q+1) + 2}\\
or,\, \frac{r}{l} \leq \frac{2}{m+1 + \frac{2}{q+1}}
\end{IEEEeqnarray*}
Since, $\frac{2}{q+1} > 0$, it follows that $\frac{2}{\frac{2}{q+1} + 1 + m} <\frac{2}{1 + m}$.
\end{IEEEproof}

\begin{theorem}
For any non-zero positive rational number $\frac{k}{n}$, and for any given set of primes $\{p_1,p_2,\ldots,p_w\}$,
there exist a sum-network which has capacity equal to $\frac{k}{n}$, and has a capacity achieving rate $\frac{k}{n}$ fractional linear network coding solution if and only if the characteristic of the finite field does not belong to the given set.
\end{theorem}
\begin{IEEEproof}
The proof is similar to that of Theorem~\ref{thm3}. Consider the sum-network $\mathcal{N}_2$ for $q=p_1\times p_2\times \cdots p_w$ and $m = 2n-1$. The proof uses the fact that the characteristic of the finite field divides $q$ if only if it belongs to the set $\{p_1,p_2,\ldots,p_w\}$. So if the characteristic does not divide $q$, it does not belong to the set $\{p_1,p_2,\ldots,p_w\}$. Consider $k$ copies of the sum-network $\mathcal{N}_2$ and name the $i^{\text{th}}$ copy as $\mathcal{N}_{2i}$. Also name the edge $e_{xj}$ in the $i^{th}$ copy as $e_{xji}$ where $1\leq x\leq m$ and $1\leq j\leq q+1$. Like $\mathcal{N}_1^\prime$ was constructed by combining $k$ copies of $\mathcal{N}_1$, let us construct $\mathcal{N}_2^\prime$ by combining $k$ copies of $\mathcal{N}_2$. Note that such a combination only merges the corresponding sources and terminals, and the rest of the sum-network remains disjoint. We first show that the sum-network has capacity equal to $\frac{k}{n}$.

Proceeding similarly as to the proof of Claim~\ref{cla3}, instead of equation (\ref{wq3}), we will have,
\begin{IEEEeqnarray*}{l}
\{Y_{e_{ijx}} | 1\leq i\leq m, 1\leq j\leq q{+}1, 1\leq x\leq k\}\\
\rightarrow\\
\bigl\{\{ X_{ixl} | 1\leq i,x\leq m, i < x, 1\leq l\leq (q+1) \},\\
\{ X_{ij} | 1\leq i\leq m, 1\leq j\leq (q+1) \}\bigr\}
\end{IEEEeqnarray*}
From this equation, as in Claim~\ref{cla3} we will get, $\frac{r}{l} \leq \frac{2k}{m+1}$. We show that $\mathcal{N}_2^\prime$ has a $(2k,m+1)$ fractional linear network coding solution if the characteristic of the finite field does not divide $q$. Since each copy of $\mathcal{N}_2$ has a $(2,m+1)$ fractional linear network coding solution over such a finite field, and there are $k$ copies of $\mathcal{N}_2$ in $\mathcal{N}_2^\prime$, consider sending the $(2i)^{\text{th}}$ and $(2i-1)^{\text{th}}$ component of the source massages through the $i^{\text{th}}$ copy $\mathcal{N}_{2i}$. Then it can be seen that any terminal can compute the required sum.

We now show that if the characteristic of the finite field divides $q$, then the rate $(2k,m+1)$ cannot be achieved in the sum-network $\mathcal{N}_2^\prime$ by using fractional linear network coding. On the contrary, say that the said rate can be achieved. Now similar to Lemma~\ref{lem1}, it can be also shown that if $\mathcal{N}_2^\prime$ has an $(r,l)$ fractional linear network coding solution then the sum-network $\mathcal{N}_2$ has an $(r,kl)$ fractional linear network coding solution. This shows that the rate $\frac{2k}{k(m+1)} = \frac{2}{m+1}$ can be achieved in $\mathcal{N}_2$ even when the characteristic divides $q$. However, this is in contradiction with Claim~\ref{cla4}. Hence, $\mathcal{N}_2^\prime$ is the required sum-network.
\end{IEEEproof}

\section{Conclusion}\label{sec4}
It is known that every rational number is the capacity of some sum-network. It is also known that for any finite or co-finite set of primes there exists a sum-network which has a capacity achieving rate $1$ linear network coding solution if and only if the characteristic of the finite field belongs to the given set. In this paper we showed that for any finite or co-finite set of prime numbers, and for any non-zero positive rational number $k/n$, there exists a sum-network which has a capacity achieving rate $k/n$ fractional linear network coding solution if and only if the characteristic of the finite field belongs to the given set of primes.

\end{document}